\documentclass[12pt]{amsart}

\usepackage[T1]{fontenc}
\usepackage{lmodern}
\usepackage{amsmath}
\usepackage{amssymb}
\usepackage{amsfonts}
\usepackage{mathtools}
\usepackage{microtype}
\usepackage[hidelinks]{hyperref}
\usepackage{graphicx}

\theoremstyle{plain}
\newtheorem{theorem}{Theorem}[section]

\newtheorem{lemma}[theorem]{Lemma}
\newtheorem{corollary}[theorem]{Corollary}
\theoremstyle{definition}
\newtheorem{definition}[theorem]{Definition}
\theoremstyle{remark}

\numberwithin{equation}{section}
\newcommand{\func}[1]{\operatorname{#1}}

\begin{document}
\title[Isospectral minimization and Ky Fan isoperimetry]{\textbf{Isospectral majorization and isoperimetric inequalities for coherent states on the Bloch sphere}}
\author{Lu\'is Daniel Abreu}
\date{}

\begin{abstract}
Let $\mathcal{P}_{N}$ be the $(N+1)$-dimensional Hilbert space of analytic
polynomials of degree at most $N$.
This is the natural environment to define $SU(2)$ (Bloch) coherent states.
Let $Q_{\rho }$ be the Husimi function of a density operator $\rho $ on $%
\mathcal{P}_{N}$. We prove an isospectral version of Lieb-Solovej
inequality: if $\rho ^{\downarrow }$ is obtained by placing the eigenvalues
of $\rho $ in decreasing order along the monomial basis, then 
\begin{equation*}
\int_{\mathbb{C}}\Phi (Q_{\rho }(z))\,dm(z)\leq \int_{\mathbb{C}}\Phi
(Q_{\rho ^{\downarrow }}(z))\,dm(z)
\end{equation*}%
for every convex function $\Phi $ on $[0,1]$. Applying the corresponding
reversed inequality to the concave function $\Phi (t)=-t\log t$ gives the
Wehrl entropy. In the process it is show that the output
state of $\rho $ under Lieb-Solovej's channel is
majorized by the output of the state $\rho
^{\downarrow }$.

As an application, among all measurable subsets of the sphere having a fixed
area, spherical caps maximize every partial sum of the eigenvalues of
Toeplitz operators with symbol $1_{\Omega}$. Equivalently, caps maximize all
Ky Fan norms, leading to the sharp inequalities, previously known for $r=1$:
\begin{equation*}
\sum_{j=1}^{r}\lambda_{j}(\Omega)
\leq
r-\sum_{k=0}^{r-1}(r-k)\binom{N+1}{k}
m(\Omega)^{k}\bigl(1-m(\Omega)\bigr)^{N+1-k}.
\end{equation*}
This implies isoperimetric inequalities for all
Schatten sums of $T_{\Omega}$, obtained without using the spherical isoperimetric inequality.
\end{abstract}

\maketitle

\section{Introduction}

Let 
\begin{equation*}
dm(z)=\frac{dA(z)}{\pi (1+|z|^{2})^{2}},\qquad z\in \mathbb{C}\text{,}
\end{equation*}%
where $dA(z)=dxdy$ is the ordinary Euclidean area measure on $\mathbb{C}$,
so that $dm(z)$ is the stereographic image of normalized area measure on the
unit sphere. For $N\geq 0$, let $\mathcal{P}_{N}$ denote the space of
analytic polynomials of degree at most $N$, with inner product 
\begin{equation*}
\langle f,g\rangle _{_{\mathcal{P}_{N}}}=(N+1)\int_{\mathbb{C}}\frac{f(z)%
\overline{g(z)}}{(1+|z|^{2})^{N}}\,dm(z)\text{.}
\end{equation*}%
The functions 
\begin{equation}
e_{n}^{(N)}(z)=\binom{N}{n}^{1/2}z^{n},\qquad 0\leq n\leq N\text{,}
\label{monomialBasis}
\end{equation}%
form an orthonormal basis. We suppress the superscript when the degree is
clear. The reproducing kernels of $\mathcal{P}_{N}$ and the associated
normalized kernels at $z$, are correspondingly given as 
\begin{equation*}
K_{N}(w,z)=(1+w\overline{z})^{N},\text{ \ \ \ \ \ \ \ \ \ \ \ \ \ }\kappa
_{N,z}(w)=\frac{(1+w\overline{z})^{N}}{(1+|z|^{2})^{N/2}}\text{.}
\end{equation*}%
We will often use ket notation \cite[Chapter 1]{Petz}.
For $f,g\in \mathcal{P}_{N}$, $|f\rangle \langle g|$ denotes the rank-$1$
operator defined by 
\begin{equation*}
\left( |f\rangle \langle g|\right) |\left. F\right\rangle =\left\langle
F,g\right\rangle _{\mathcal{P}_{N}}|\left. f\right\rangle ,\text{ \ \ \ \ }%
F\in \mathcal{P}_{N}\text{.}
\end{equation*}%
The normalized kernels satisfy, in this notation, the resolution of the
identity 
\begin{equation}
(N+1)\int_{\mathbb{C}}|\kappa _{N,z}\rangle \langle \kappa _{N,z}|\,dm(z)=I_{%
\mathcal{P}_{N}}\text{.}  \label{eq:resolution}
\end{equation}%
A unit vector $(u,v)\in \mathbb{C}^{2}$, $|u|^{2}+|v|^{2}=1$ (named as a 
\emph{spinor} in physics language) defines a $SU(2)$ coherent state by 
\begin{equation}
\kappa _{N,(u,v)}(w),=\sum_{n=0}^{N}\binom{N}{n}%
^{1/2}u^{N-n}v^{n}e_{n}^{(N)}(w)\text{.}  \label{eq:spincoherent}
\end{equation}%
If we choose%
\begin{equation*}
u=\frac{1}{\sqrt{1+\left\vert z\right\vert ^{2}}}\text{, \ \ \ \ \ \ \ \ }v=%
\frac{\overline{z}}{\sqrt{1+\left\vert z\right\vert ^{2}}}\text{,}
\end{equation*}%
then $\kappa _{N,(u,v)}(w)=\kappa _{N,z}(w)$. These coherent states in $%
\mathcal{P}_{N}$ are usually called Bloch or $SU(2)$ coherent states. Their
phase space is the Riemann sphere, written in a stereographic coordinate $%
z\in \mathbb{C}$.

 A density operator is a positive trace-class operator $%
\rho $ with $\func{Tr}(\rho )=1$ \cite[Section 2.1]{Petz}.  Since $\dim \mathcal{P}_{N}=N+1$, every density operator in $%
\mathcal{P}_{N}$ satisfies $1\leq Rank(\rho )\leq N+1$.

To give an adequate perspective for the work presented in this paper, we
first present an outline of the related results by Lieb and Solovej achieved
in the fundamental paper \cite{LiebSolovej}, where, in particular, a
previous conjecture of Lieb \cite{Lieb1978} for the Wehrl entropy of $SU(2)$
coherent states was confirmed. More precisely, Lieb and Solovej proved in 
\cite[Theorem~2.1]{LiebSolovej} that, for every density operator $\rho $ on $%
\mathcal{P}_{N}$, every convex function $f:[0,1]\rightarrow \mathbb{R}$, and
every rank-$1$ coherent-state projection $\rho _{\mathrm{coh}}$, 
\begin{equation}
\int_{\mathbb{C}}f(Q_{\rho }(z))\,dm(z)\leq \int_{\mathbb{C}}f(Q_{\rho _{%
\mathrm{coh}}}(z))\,dm(z)\text{.}  \label{eq:LShusimi}
\end{equation}%
For concave $f$, the inequality is reversed. Taking $f(t)=-t\log t$ proves
the Bloch coherent-state Wehrl conjecture. This result is the starting point
of our investigations.

Our first observation is that the comparison in inequality (\ref{eq:LShusimi}%
) does not preserve the spectrum of $\rho $, since a general density
operator $\rho $ or arbitrary rank is compared with the coherent pure state $%
\rho _{\mathrm{coh}}$. This can be illustrated with the simple example of a
rank-$r$ projection $P$ and the associated density operator $\rho =P/r$,
then the spectrum of $\rho $ is 
\begin{equation}
\left( \frac{1}{r},\ldots ,\frac{1}{r},0,\ldots ,0\right) \text{.}
\label{spectumInputLS}
\end{equation}%
In contrast, $\rho _{\mathrm{coh}}$ is a pure state, which, by definition,
is a positive trace-one operator of rank one. Thus, the spectrum (\ref%
{spectumInputLS}) is compared with the pure spectrum $(1,0,\ldots ,0)$. The
result therefore identifies a global maximizer for density operators with
arbitrary spectrum, but it does not identify a density operator maximizing a
prescribed higher rank spectral class. Such maximizers are important for the
understanding of the quantum features of density operators, since they
represent the most classical state of the system. In other words, they can
be regarded as\emph{\ }Bloch coherent state version of (rank-$r$) \emph{%
quantum Gaussian states}, which are the solutions of the constrained minimum
output conjectures. We refer to \cite[Conjecture V.1]{DPTGA} for a topical
review of these conjectures for the Fock states.

The terminology `quantum Gaussian states' has a clear intuition behind. Real
and complex Gaussian functions play a distinguished role among rank-$1$
operators in Euclidean settings: they minimize Heisenberg uncertainty,
describe Wigner positivity \cite{Hudson}, and, by Lieb's theorem \cite%
{Lieb1978}, maximize the integral of convex functions of the Husimi
function. Maximization relations of this nature, comparing averaged
amplitudes of coherent state transforms of a given density operator and its
minimizer (the corresponding Gaussian operator of the system), can be
leveraged to compute the \emph{Gausssianity deficit }of the quantum state
described by the density operator. Such deficits are of paramount importance
in quantum science, because they reveal, in a quantitative form, the quantum
features of the given state \cite{DPTGA,Stellar}. This becomes of particular
importance in the Bloch phase space, where the Wigner negativity exhibited
by some coherent states, severely limits its use as a measure of quantum
resources \cite{WignerSU2}. In this mathematical work, we will not dwell
deeper into this direction, but rather emphasize it as the physical
motivation for the presented results, which might be useful for the
investigation of Gaussianity deficits of Bloch coherent states. In this
sense, the Lieb--Solovej inequality provides a natural deficit relative to
coherent states. Our comparison yields a finer deficit adapted to the
prescribed-spectrum problem. Indeed, comparing (\ref{eq:LShusimi}) with our
inequality in Theorem \ref{thm:sphericaldepalma}, gives, for a convex
function $\Phi :[0,1]\rightarrow \mathbb{R}$, 
\begin{equation}
\int_{\mathbb{C}}\Phi (Q_{\rho }(z))\,dm(z)\leq \int_{\mathbb{C}}\Phi
(Q_{\rho ^{\downarrow }}(z))\,dm(z)\leq \int_{\mathbb{C}}\Phi (Q_{\rho _{%
\mathrm{coh}}}(z))\,dm(z)\text{,}  \label{3Husimis}
\end{equation}%
where $\rho ^{\downarrow }$ is obtained by placing the eigenvalues of $\rho $
in decreasing order along the monomial weight basis. Another motivation
arises from the recent work of Aschieri, Ruba and Solovej \cite{Aschieri},
where it is shown that there exist $SU(2)$-equivariant channels whose
minimizers are not coherent states, therefore suggesting that the landscape
of such minimizers can be quite vast. While the channel we will use is the
same as in Lieb-Solovej, and well known to be $SU(2)$-equivariant, it also
contributes to this diversity of minimizers, by finding one constrained to
have the same spectrum of the original density operator. Indeed, from
Theorem \ref{thm:iterated} and \cite[Theorem 2.4]{LiebSolovej} we obtain a
majorization relation in the same order of (\ref{3Husimis}) 
\begin{equation*}
\mathcal{C}_{N}(\rho )\prec \mathcal{C}_{N}(\rho ^{\downarrow })\prec 
\mathcal{C}_{N}(\rho _{\mathrm{coh}})\text{.}
\end{equation*}

For the rank-$1$ case of $SU(2)$ coherent states, sharp results,
characterizing the maximizers, have been obtained recently. The sharp local
rank-one statement, including the characterization of coherent states and
spherical caps as the maximizers, was obtained by Frank \cite{Frank} and,
independently, by Kulikov, Nicola, Ortega-Cerd\`{a} and Tilli \cite%
{Frank,KNOT}. Garc\'{\i}a-Ferrero and Ortega-Cerd\`{a} later proved a
quantitative stability theorem in one complex dimension \cite{GFOC}. These
results settle the optimizer and stability questions in the rank-$1$ case
for the space $\mathcal{P}_{N}$. Related results have been also obtained for
Fock coherent states \cite{Lieb1978,NicolaTilli2022,GGRT} and for $SU(1,1)$ 
\cite{Kulikov,RamTilli,ProcLMS}, and provide optimal bounds for the
sparsity/concentration trade-off, previously studied using non-sharp large
sieve type methods \cite{AbrSpeck,Speck}. We also observe that a
maximization inequality for general density operators in a more general
coherent states context was obtained in \cite[Section 2.3]{Frank}, but, in
the spirit of Lieb-Solovej inequality, the maximizer is a rank-$1$ operator.
None of the papers contained in this significant body of work addresses the
fixed spectrum comparison problem for higher rank density operators. 

The main goal of this paper is to fill this gap in the mathematical theory
developed for higher rank $SU(2)$ coherent states: for each given density
operator $\rho $, we will determine the corresponding minimizer state, in
terms of the Husimi function average involved in Lieb-Solovej inequality (%
\ref{eq:LShusimi}). In the process, we prove a fixed trace majorization
theorem for the Lieb and Solovej amplification channel \cite[Theorems~2.1]%
{LiebSolovej}, and obtain, as by-product of Husimi maximization, the sharp
isoperimetric inequality for the Ky Fan norm of Toeplitz operators, which
automatically implies sharp isoperimetric inequalities for all unitarily
invariant matrix norms, in particular for Schatten and quasi Schatten norms
for every $0<p\leq \infty $.

\subsection{Main results}

Our results deal with the case of a general rank-$r$ operator, and the
arguments start by developing a fixed spectrum output majorization theory
for the Lieb-Solovej amplification channel, which we state as an independent
result, Theorem \ref{thm:iterated}. This retains the flat rank-$r$ spectrum
and controls every accumulated eigenvalue sum. The most significant result
in this direction is Theorem \ref{thm:sphericaldepalma}. It preserves the
full spectrum of the density operator and therefore supplies information
unavailable from the rank-one coherent-state inequality. To our knowledge,
this is the first higher-rank result that identifies a minimizer under a
prescribed-spectrum constraint: within every fixed spectral class, the
passive state is a minimizer. In particular, it retains the flat rank-$r$
spectrum and this will be combined with the Ky-Fan maximum principle for
Toeplitz operators $T_{\Omega }$ of density functions with a flat symbol $%
1_{\Omega }$, in the proof of our second main result, which controls every
accumulated eigenvalue sum for $T_{\Omega }$, therefore leading to a Ky Fan
isoperimetric inequality which automatically implies isoperimetric
inequalities for all unitarily invariant norms, in particular, for all
Schatten norms and quasi-norms of $T_{\Omega }$. Before presenting our main
results we introduce the concept of $SU(2)$ passive rearrangements.

\begin{definition}
Consider a density operator $\rho $ in $\mathcal{P}_{N}$, with $Rank(\rho )=r
$, and its Hilbert-Schmidt decomposition as a convex linear combination of
the rank-$1$ projections over its eigenfunctions $F_{n}\in \mathcal{P}_{N}$,
and its eigenvalues $p_{n}$: 
\begin{equation*}
\rho =\sum_{n=0}^{r-1}p_{n}|F_{n}\rangle \langle F_{n}|\qquad
\sum_{n=0}^{r-1}p_{n}=1.
\end{equation*}%
Write the eigenvalues of $\rho $ in non-increasing order, $p_{0}\geq
p_{1}\geq \cdots \geq p_{r-1}\geq 0,$ and let $e_{0},e_{1},\ldots ,e_{N}$ be
the normalized monomial basis defined by (\ref{monomialBasis}). \emph{The }$%
SU(2)$\emph{\ passive rearrangement} of $\rho $ is 
\begin{equation*}
\rho ^{\downarrow }=\sum_{n=0}^{r-1}p_{n}|e_{n}\rangle \langle e_{n}|.
\end{equation*}%
A density operator is passive if it is equal to its passive rearrangement.
\end{definition}

Our central result is the following fixed spectrum version of Lieb-Solovej
inequality (\ref{eq:LShusimi}).

\begin{theorem}
\label{thm:sphericaldepalma} Let $\rho $ be a density operator on $\mathcal{P%
}_{N}$, and let $\rho ^{\downarrow }$ be its passive rearrangement. Then,
for every convex function $\Phi :[0,1]\rightarrow \mathbb{R}$, 
\begin{equation}
\int_{\mathbb{C}}\Phi (Q_{\rho }(z))\,dm(z)\leq \int_{\mathbb{C}}\Phi
(Q_{\rho ^{\downarrow }}(z))\,dm(z).  \label{eq:sphericaldepalma}
\end{equation}
\end{theorem}

This compares $\rho $ with an operator with the same spectrum, changing only
its eigenvectors. In particular, the operator whose spectrum is (\ref%
{spectumInputLS}) is compared with a state with the same spectrum. Define
the Bloch Wehrl entropy by 
\begin{equation}
S_{W}(\rho )=-(N+1)\int_{\mathbb{C}}Q_{\rho }(z)\log Q_{\rho }(z)\,dm(z),
\label{eq:BlochWehrlEntropy}
\end{equation}%
where $0\log 0$ is understood to be zero. Applying the concave form of
Theorem \ref{thm:sphericaldepalma} to the concave function $\Phi (t)=-t\log
t,$with $\Phi (0)=0$ leads to the fixed spectrum minimization for the Wehrl
entropy of Bloch states.

\begin{corollary}
Let $\rho $ be a density operator on $\mathcal{P}_{N}$, and let $\rho
^{\downarrow }$ be obtained by placing the eigenvalues of $\rho $ in
decreasing order along the monomial basis. Then 
\begin{equation}
S_{W}(\rho )\geq S_{W}(\rho ^{\downarrow }).  \label{eq:fixedspectrumWehrl}
\end{equation}%
Consequently, among all density operators on $\mathcal{P}_{N}$ having the
same spectrum as $\rho $, the passive state $\rho ^{\downarrow }$ is a
minimizer of the Bloch Wehrl entropy.
\end{corollary}

The proof involves the following fixed spectrum majorization theorem for the
Lieb-Solovej iterated channel, which is defined in Section 2 in a slightly
different form. The result may have independent interest and is stated here
as a Theorem.

\begin{theorem}
\label{thm:iterated} For every density operator $\rho $ on $\mathcal{P}_{N}$
and every $M\geq N$, 
\begin{equation}
\mathcal{C}_{N\rightarrow M}(\rho )\prec \mathcal{C}_{N\rightarrow M}(\rho
^{\downarrow }).  \label{eq:iteratedmajorization}
\end{equation}%
The state on the right is passive in $\mathcal{P}_{M}$. 
\end{theorem}

Finally, we will deduce from Theorem \ref{thm:sphericaldepalma} an
isoperimetric inequality for the Ky Fan norm of the Toeplitz operator $%
T_{\Omega }$ defined as 
\begin{equation*}
T_{\Omega }=(N+1)\int_{\Omega }|\kappa _{N,z}\rangle \langle \kappa
_{N,z}|\,dm(z)\text{.}
\end{equation*}%
This is our final main result. It provides a spherical analogue of the
inequality for the Fock space, recently obtained by the author \cite%
{AbreuKyFan}, after a conjecture in \cite%
{NicolaRiccardiTilli2025,NicolaRiccardiTilli2026}, where the inequality was
proved with radial assumptions on $\Omega $ and the existence of maximizers
proved for the unrestricted case. In the spherical setting considered in
this paper, the Ky Fan isoperimetric Theorem reads as follows.

\begin{theorem}
\label{thm:sphericalkyfan} Let $\Omega $ be a measurable subset of the
sphere and let $D_{\alpha }$ be a spherical cap with 
\begin{equation*}
m(\Omega )=m(D_{\alpha })=\alpha .
\end{equation*}%
Then, for every $1\leq r\leq N+1$, 
\begin{equation}
\sum_{j=1}^{r}\lambda _{j}(\Omega )\leq \sum_{j=1}^{r}\lambda _{j}(D_{\alpha
})\text{.}  \label{eq:kyfanmain}
\end{equation}%
Equivalently,%
\begin{equation*}
\sum_{j=1}^{r}\lambda _{j}(\Omega )\leq r-\sum_{k=0}^{r-1}(r-k)\binom{N+1}{k}%
m(\Omega )^{k}(1-m(\Omega ))^{N+1-k}\text{.}
\end{equation*}%
The right side is attained by the passive projection onto $\func{span}%
\{e_{0},\ldots ,e_{r-1}\}$, after a spherical rotation placing its pole at
the center of the cap.
\end{theorem}

Theorem \ref{thm:sphericalkyfan} provides isoperimetric inequalities for all
Schatten norms of the operator $T_{\Omega }$. We proceed with the
definitions of Schatten norms and quasi norms. For $0<p<\infty $, the
Schatten $p$-sum of a positive compact operator $T$ is 
\begin{equation*}
\Vert T\Vert _{S^{p}}=\left( \sum_{j=1}^{N+1}\lambda _{j}(T)^{p}\right)
^{1/p}.
\end{equation*}%
For $p\geq 1$, $\Vert T\Vert _{S^{p}}$ is a norm and for $0<p<1$, $\Vert
T\Vert _{S^{p}}$ is a quasinorm. Since the sums are finite, these quantities
are always defined for every operator on $\mathcal{P}_{N}$. We also write $%
\Vert T\Vert _{S^{\infty }}=\lambda _{1}(T)=\Vert T\Vert _{\mathrm{op}}.$

\begin{corollary}
\label{AllNorms}For $1<p\leq \infty $, 
\begin{equation*}
\Vert T_{\Omega }\Vert _{\mathcal{S}^{p}}\leq \Vert T_{D_{\alpha }}\Vert _{%
\mathcal{S}^{p}},
\end{equation*}%
while for $0<p<1$ the corresponding Schatten quasinorm inequality is
reversed. For $p=1$ both sides are $(N+1)\alpha $. More generally, 
\begin{equation*}
\Vert T_{\Omega }\Vert _{\mathfrak{I}}\leq \Vert T_{D_{\alpha }}\Vert _{%
\mathfrak{I}}
\end{equation*}%
for every unitarily invariant matrix norm $\Vert \cdot \Vert _{\mathfrak{I}}$%
. Every Ky Fan norm satisfies the sharper partial-sum inequality (\ref%
{eq:kyfanmain}).
\end{corollary}

In contrast with the situation in the Fock space, where majorization theory
and the Karamata-type inequalities from \cite{NicolaRiccardiTilli2025} have
been used to deduce the Schatten inequalities from the Ky Fan isoperimetric
inequality, the above Corollary follows automatically from Theorem \ref%
{thm:sphericalkyfan}. Indeed, since the dimension of the ambient space is
finite, the Toeplitz operators have finite rank. Thus, all we need is the Ky
Fan dominance Theorem \cite[{\small Section IV.2, Theorem IV.2.2}]{Bhatia},
which assures that, given two matrices $A,B\in M_{n}(\mathbb{C})$, their
singular values are dominated in the weak sense 
\begin{equation*}
\sum_{j=1}^{r}s_{j}(A)\leq \sum_{j=1}^{r}s_{j}(B)\text{, \ \ \ \ \ \ \ \ \ \
\ \ \ \ \  }r=1,...,N+1\text{.}
\end{equation*}%
if and only if $\Vert A\Vert \leq \Vert B\Vert $ for every unitarily
invariant matrix norm. Since $T_{\Omega }$ is compact and positive, its
eigenvalues coincide with its singular values, and this gives Corollary \ref%
{AllNorms}. The result for $p>1$ follows from a standard application of
Karamata's inequality to the convex function $t^{p}$ and for $0<p<1$ to the
concave function $t^{p}$. For $p=1$ the inequality follows from $\func{Tr}%
(T_{\Omega })=\func{Tr}(T_{D_{\alpha }})=\alpha (N+1)$ and for $p=\infty $
it is the case $r=1$ of (\ref{eq:kyfanmain}).

We observe that \emph{no isoperimetric inequality is being used here, the
spherical-cap geometry controlling the norms follows from the fixed-spectrum majorization}.
The argument based on (\ref{eq:kyfanmain}) together with its proof, explains
why the same circular geometry controls all unitarily invariant norms: \emph{the extremal density operator obtained by these
methods is axially invariant on }$\mathcal{P}_{N}$,\emph{\ while its Husimi
function is invariant under the corresponding rotations of the Bloch sphere. 
}It is\emph{\ }the resulting radial structure that identifies spherical
caps---or stereographic discs---as the extremal sets for every Ky Fan sum.
In the proof of Theorem \ref{thm:sphericalkyfan}, this geometric principle
has an explicit form: the eigenvalues $\lambda _{j}(D_{\alpha })$ are
associated with the eigenfunctions (\ref{monomialBasis}). The radial
structure of the eigenfunctions leads to orthogonality in any spherical cap,
which automatically diagonalizes $T_{\Omega }$

The Hilbert--Schmidt case $p=2$ can also be proved directly from spherical
rearrangement applied to the squared coherent-state kernel, the case $p=1$
corresponds to the trace of the operator and the case $p=\infty $ can be
obtained from Lieb and Solovej and corresponds to the spherical Faber--Krahn
inequality obtained by Frank \cite{Frank} and Kulikov, Nicola, Ortega-Cerd%
\`{a} and Tilli \cite{Frank,KNOT}. Nevertheless, and to complete the attempt
of delimitating the current state of the art in these problems, we emphasize
that our methods do not characterize all the equality cases, as \cite{Frank}
and \cite{Frank,KNOT} do for the rank-$1$ case. 

Our results are inspired and effectively use ideas with diverse origins. The
central concept governing the argument is exactly the same amplification
channel for Bloch states introduced by Lieb and Solovej \cite[Theorems~2.1]%
{LiebSolovej}, which is first introduced in the polynomial setting
described in the introductory paragraphs . A fixed spectrum majorization
theory will be developed specifically for this channel. This is
inspired in fixed output majorization results for Fock states by de Palma,
Trevisan and Giovannetti \cite{DPTG}, leveraged by de Palma \cite%
{DePalma2018} to obtain the Husimi fixed spectrum maximization for Fock
states, under continuity assumptions (see also the follow up \cite{DPTGA},
which includes an extensive topical review). This result was recently
used by the author to prove the isoperimetric inequality for the Ky Fan and
Schatten norms of time-frequency localization operators \cite{AbreuKyFan}.
Theorem \ref{thm:sphericalkyfan} shows that similar results hold for Bloch $SU(2)$ states. 

\subsection{Outline}

The presentation is organized as follows. Section 2 introduces the concepts involved in the main results and their proofs: the first subsection introduces Husimi functions, lower symbols, and the Berezin transform, the second presents an exposition of some useful fundamental concepts from quantum information theory, namely quantum channels, Kraus representations, and Ky Fan maximum principles. This is aimed to make the presentation self-contained for a mathematical audience. The third and fourth subsections introduce the spherical amplification channel in the polynomial space  and make the identification with the Lieb-Solovej channel. In section 3, the
fixed spectrum majorization results for output of channels, of which the most general is
Theorem \ref{thm:iterated}, are proved. Section 4 proves Theorem \ref{thm:sphericaldepalma}. The paper is concluded with
section 5, where Theorem \ref{thm:sphericalkyfan} is demonstrated.

\section{The Husimi function and the spherical amplification channels}

\subsection{Density operators, Husimi functions and lower symbols}

A density operator is a positive trace-class operator $\rho $ with $\func{Tr}%
(\rho )=1$ \cite[Section 2.1]{Petz}. Thus, density operators are compact and
self-adjoint and the spectral theorem for self-adjoint operators applies,
assuring that $\rho $ always has eigenvalues (even in the infinite
dimensional case, which is not required here). Thus, if $\rho $ is a density
operator, the singular values equal its eigenvalues, $0\leq \rho \leq 1$ and
the eigenvalues are all non-negative. Since $\dim \mathcal{P}_{N}=N+1$,
every density operator in $\mathcal{P}_{N}$ satisfies%
\begin{equation*}
1\leq Rank(\rho )\leq N+1\text{.}
\end{equation*}%
Let $\rho $ be a density operator on $\mathcal{P}_{N}$. Its spherical Husimi
function is defined by the action 
\begin{equation}
Q_{\rho }(z)=\langle \kappa _{N,z},\rho \kappa _{N,z}\rangle _{N}
\label{eq:husimi}
\end{equation}%
and always satisfies 
\begin{equation}
0\leq Q_{\rho }\leq 1,\qquad \int_{\mathbb{C}}Q_{\rho }(z)\,dm(z)=\frac{1}{%
N+1}.  \label{eq:husiminormalization}
\end{equation}%
The Husimi function is sometimes normalized as $(N+1)Q_{\rho }$ to define an
ordinary probability density, but the above normalization, where $0\leq
Q_{\rho }\leq 1$, will be more convenient for our purposes. 

For an operator $X\in End(\mathcal{P}_{M})$, define its coherent-state lower
symbol, or covariant symbol (for density operators this becomes the Husimi
function), as%
\begin{equation}
Q_{X}(z)=\left\langle X\kappa _{M,z},\kappa _{M,z}\right\rangle _{\mathcal{P}%
_{M}}  \label{lowersymbol}
\end{equation}

For a bounded function $a$ on the sphere, define the spherical Toeplitz
operator with symbol $a$ on $\mathcal{P}_{M}$ by 
\begin{equation}
\mathcal{T}_{M}(a)=(M+1)\int_{\mathbb{C}}a(z)|\kappa _{M,z}\rangle \langle
\kappa _{M,z}|\,dm(z).  \label{eq:toeplitzquant}
\end{equation}%
Its lower symbol is the Berezin transform 
\begin{equation}
(\mathcal{B}_{M}a)(z)=(M+1)\int_{\mathbb{C}}|\langle \kappa _{M,z},\kappa
_{M,w}\rangle |^{2}a(w)\,dm(w).  \label{eq:berezin}
\end{equation}%
This is the coherent-state quantization associated with the spin
representation \cite{Berezin}.

We will prove several majorization results and use a notation that is
convenient to define precisely. For Hermitian matrices $A$ and $B$ of the
same size, we write $A\prec B$ when the decreasing eigenvalue vector of $A$
is majorized by that of $B$. Thus 
\begin{equation*}
\sum_{j=1}^{s}\lambda _{j}(A)\leq \sum_{j=1}^{s}\lambda _{j}(B)
\end{equation*}%
for every partial sum, with equality for the total sums. We will use the
following convention everywhere in the paper. Whenever an empty sum appears
from setting $s=1$ in 
\begin{equation*}
\sum_{j=0}^{s-2}p_{j}
\end{equation*}%
we declare such empty sum to be $0$.

\subsection{Background on quantum channels and Kraus representations}

We will use some basic concepts and results about quantum channels, which
can be found in \cite{Petz,Trevisan}.\ 

\textbf{Channels}

A channeling transformation is a map $\mathcal{C}:\mathcal{B}\left( \mathcal{%
H}_{1}\right) \mathcal{\rightarrow B(H}_{2}\mathcal{)}$ which transform
density operators from $\mathcal{H}_{1}$ (the input) into density operators
in $\mathcal{H}_{2}$ (the output). In finite dimensions, a \emph{quantum
channel} is a linear map $\Phi :\mathcal{B}(\mathcal{H}_{1})\rightarrow 
\mathcal{B}(\mathcal{H}_{2})$ that is \emph{completely positive and trace
preserving}. Complete positivity means that 
\begin{equation*}
I_{k}\otimes \Phi :M_{k}(\mathbb{C})\otimes \mathcal{B}(\mathcal{H}%
_{1})\rightarrow M_{k}(\mathbb{C})\otimes \mathcal{B}(\mathcal{H}_{2})\text{,%
}
\end{equation*}%
defined on elementary tensors by%
\begin{equation*}
(I_{k}\otimes \Phi )(A\otimes X)=A\otimes \Phi (X)
\end{equation*}%
maps positive operators to positive operators for every $k\geq 1$, where $%
I_{k}$ denotes the identity map on $M_{k}(\mathbb{C})$. This is a stronger
requirement than positivity of $\Phi $, which is worth justifying:
positivity of $\Phi $ is enough when the operator $X$ appears alone, in a
positive product operator $A\otimes X$, or in a sum of positive product
operators. However, a general positive operator $Y$ on the tensor-product
space need not be a sum of positive product operators, and positivity of $%
\Phi $ does not by itself guarantee that%
\begin{equation*}
\left( I_{k}\otimes \Phi \right) \left( Y\right) \geq 0.
\end{equation*}%
Complete positivity requires this inequality to hold for every positive $Y$
in the tensor product space (the composite system). We note in passing that,
after normalization to trace one, positive operators that cannot be written
as sums of positive product operators are precisely the entangled states of
the composite system.

\textbf{Kraus representations}

A Kraus representation is an operator-sum formula 
\begin{equation}
\Phi (X)=\sum_{\ell }K_{\ell }XK_{\ell }^{\ast }\text{,}
\label{eq:kraus-representation}
\end{equation}%
where $K_{\ell }$ are the Kraus operators of the representation. Such a
formula is completely positive and trace preserving if and only if 
\begin{equation}
\sum_{\ell }K_{\ell }^{\ast }K_{\ell }=I_{\mathcal{H}_{1}}.
\label{eq:kraus-trace-preserving}
\end{equation}%
These standard facts, including the operator-sum characterization, are given
in \cite[Theorems~2.1--2.2, pp.~15--16; Section~7.1, pp.~91--92]{Petz}. The
adjoint map $\Phi ^{\ast }:\mathcal{B}(\mathcal{H}_{2})\rightarrow \mathcal{B%
}(\mathcal{H}_{1})$ is determined by trace duality, 
\begin{equation}
\func{Tr}\left( Y\Phi (X)\right) =\func{Tr}\left( \Phi ^{\ast }(Y)X\right) .
\label{eq:adjoint-definition}
\end{equation}%
If (\ref{eq:kraus-representation}) holds, then $\Phi ^{\ast }(Y)=\sum_{\ell
}K_{\ell }^{\ast }YK_{\ell }$. Consequently, \emph{a trace-preserving map
has a unital adjoint }\cite[Section 3.5]{Trevisan}:%
\begin{equation*}
\Phi ^{\ast }(I)=I.
\end{equation*}
In particular, positivity gives 
\begin{equation}
0\leq Y\leq I\quad \Longrightarrow \quad 0\leq \Phi ^{\ast }(Y)\leq I.
\label{eq:adjoint-contraction}
\end{equation}

\textbf{Ky Fan maximum principle}

For a Hermitian $d\times d$ matrix $A$, define its $s$th Ky Fan sum by 
\begin{equation}
K_{s}(A)=\sum_{j=0}^{s-1}\lambda _{j}(A),\qquad 1\leq s\leq d.
\label{eq:kyfan-sum}
\end{equation}%
Ky Fan's maximum principle tells 
\begin{equation}
K_{s}(A)=\max_{\substack{ R=R^{\ast }=R^{2} \\ \func{rank}R=s}}\func{Tr}(RA).
\label{eq:kyfan-principle}
\end{equation}%
See \cite[Theorem~1, formula~(4)]{Fan} and \cite[Theorem~11.13, p.~178]{Petz}%
. 

\subsection{The spherical amplification channel}

We
start by defining a map $\mathcal{C}_{N}:\mathcal{B}\left( \mathcal{P}%
_{N}\right) \mathcal{\rightarrow B}\left( \mathcal{P}_{N+1}\right) $ by its 
\emph{Kraus representation} 
\begin{equation}
\mathcal{C}_{N}(\rho )=A_{N,0}\rho A_{N,0}^{\ast }+A_{N,1}\rho A_{N,1}^{\ast
}\text{,}  \label{eq:channel}
\end{equation}%
where the Kraus operators $A_{N,0},A_{N,1}:\mathcal{P}_{N}\rightarrow 
\mathcal{P}_{N+1}$ are defined as 
\begin{align}
A_{N,0}e_{n}^{(N)}& =\left( \frac{N+1-n}{N+2}\right) ^{1/2}e_{n}^{(N+1)}%
\text{,}  \label{eq:A0} \\
A_{N,1}e_{n}^{(N)}& =\left( \frac{n+1}{N+2}\right) ^{1/2}e_{n+1}^{(N+1)}%
\text{.}  \label{eq:A1}
\end{align}

By defining the map by its Kraus representation (\ref{eq:channel}) it is
automatically completely positive. Moreover, equations (\ref{eq:A0})--(\ref%
{eq:A1}) give 
\begin{equation}
A_{N,0}^{\ast }A_{N,0}+A_{N,1}^{\ast }A_{N,1}=I_{\mathcal{P}_{N}},
\label{IdKrauss}
\end{equation}%
since%
\begin{equation*}
\left( A_{N,0}^{\ast }A_{N,0}+A_{N,1}^{\ast }A_{N,1}\right)
e_{n}^{(N)}=e_{n}^{(N)}\text{.}
\end{equation*}%
Therefore $\mathcal{C}_{N}$ is completely positive and trace preserving and
hence a quantum channel. For $M>N$, we define the iterated channel $\mathcal{%
C}_{N\rightarrow M}:\mathcal{B}\left( \mathcal{P}_{N}\right) \mathcal{%
\rightarrow B}\left( \mathcal{P}_{M}\right) $ by 
\begin{equation}
\mathcal{C}_{N\rightarrow M}=\mathcal{C}_{M-1}\circ \cdots \circ \mathcal{C}%
_{N}.  \label{itChannel}
\end{equation}

\subsection{Identification with the Lieb-Solovej channel}

In this subsection we show that the channel (\ref{itChannel}) is the same
channel used in \cite[(21)]{LiebSolovej}, after translating spin notation
into polynomial-degree notation and fixing the normalization.  Set 
\begin{equation}
J=\frac{N}{2},\qquad K=\frac{M}{2},\qquad k=2(K-J)=M-N.
\label{eq:spinparameters}
\end{equation}%
The spaces in the two notations are identified by 
\begin{equation}
\mathcal{P}_{N}\simeq \func{Sym}^{N}(\mathbb{C}^{2})\simeq \mathcal{H}%
_{J},\qquad \mathcal{P}_{M}\simeq \func{Sym}^{M}(\mathbb{C}^{2})\simeq 
\mathcal{H}_{K},  \label{eq:spaceidentifications}
\end{equation}%
and similarly 
\begin{equation}
\mathcal{P}_{M-N}\simeq \mathcal{H}_{K-J}.
\label{eq:auxiliaryspaceidentification}
\end{equation}%
Lieb and Solovej first define a channel $\Phi ^{k}$ and then introduce the
modified channel  
\begin{equation}
\widetilde{\Phi }^{k}(\rho )=\Phi ^{k}(U_{J}\rho U_{J}^{-1}),
\label{eq:LSmodifiedchannel}
\end{equation}%
where $U_{J}$ is the spin-reversing antiunitary operator \cite[(16)]%
{LiebSolovej}. For $K\geq J$, they prove   
\begin{equation}
\widetilde{\Phi }^{k}(\rho )=\frac{2J+1}{2K+1}\mathbf{P}_{K}\left( I_{%
\mathcal{H}_{K-J}}\otimes \rho \right) \mathbf{P}_{K},
\label{eq:LSprojectionchannel}
\end{equation}%
where $\mathbf{P}_{K}$ is the orthogonal projection of $\mathcal{H}%
_{K-J}\otimes \mathcal{H}_{J}$ onto its maximal-spin component $\mathcal{H}%
_{K}$ \cite[(17)]{LiebSolovej}.  Under the identifications (\ref%
{eq:spaceidentifications})-- (\ref{eq:auxiliaryspaceidentification}), the
coherent vector in this maximal-spin component is 
\begin{equation}
\kappa _{M,\xi }=\kappa _{M-N,\xi }\otimes \kappa _{N,\xi },
\label{eq:coherenttensor}
\end{equation}%
where $\xi \in \mathbb{C}^{2}$ is a unit vector. In particular, $\mathbf{P}%
_{K}\kappa _{M,\xi }=\kappa _{M,\xi }$.

Therefore, (\ref{eq:LSprojectionchannel}) gives 
\begin{align}
Q_{\widetilde{\Phi }^{,k}(X)}(\xi )& =\left\langle \widetilde{\Phi }%
^{,k}(X)\kappa _{M,\xi },\kappa _{M,\xi }\right\rangle   \notag \\
& =\frac{N+1}{M+1}\left\langle \left( I_{\mathcal{P}_{M-N}}\otimes X\right)
\left( \kappa _{M-N,\xi }\otimes \kappa _{N,\xi }\right) ,\kappa _{M-N,\xi
}\otimes \kappa _{N,\xi }\right\rangle   \notag \\
& =\frac{N+1}{M+1}\left\langle X\kappa _{N,\xi },\kappa _{N,\xi
}\right\rangle   \notag \\
& =\frac{N+1}{M+1}Q_{X}(\xi ).  \label{eq:LSlowersymbol}
\end{align}%
We now compare this with the iteration of the one-step channel. Directly
from (\ref{eq:channel})--(\ref{eq:A1}), 
\begin{equation}
Q_{\mathcal{C}_{L}(X)}(\xi )=\frac{L+1}{L+2}Q_{X}(\xi ).
\label{eq:onesteplowersymbol}
\end{equation}%
Applying (\ref{eq:onesteplowersymbol}) successively for $L=N,N+1,\ldots ,M-1$%
, we obtain 
\begin{align}
Q_{\mathcal{C}_{N\rightarrow M}(X)}(\xi )& =\frac{N+1}{M+1}Q_{X}(\xi ).  \label{eq:iteratedlowersymbol}
\end{align}%
Equations (\ref{eq:LSlowersymbol}) and (\ref{eq:iteratedlowersymbol}) show
that 
\begin{equation}
Q_{\mathcal{C}_{N\rightarrow M}(X)}=Q_{\widetilde{\Phi }^{M-N}(X)}.
\label{eq:equalsymbols}
\end{equation}%
Both operators act on $\mathcal{P}_{M}$. By Lemma \ref{lem:berezincalculus},
the lower-symbol map on $\func{End}(\mathcal{P}_{M})$ is injective. As a
result, (\ref{eq:equalsymbols}) implies 
\begin{equation}
\mathcal{C}_{N\rightarrow M}(X)=\widetilde{\Phi }^{M-N}(X)
\end{equation}%
for every $X\in \func{End}(\mathcal{P}_{N})$. Consequently, 
\begin{equation}
\mathcal{C}_{N\rightarrow M}=\widetilde{\Phi }^{M-N}\text{.}
\label{eq:channelLSidentification}
\end{equation}

\section{Fixed spectrum majorization of channel outputs}

\subsection{Preliminary results}

The following result computes the action of $\mathcal{C}_{N}$ on the
identity and in the lower symbol of a general operator. We already know from
(\ref{IdKrauss}) that  $\mathcal{C}_{N}$ is completely positive and trace
preserving and therefore a channel. We now compute its action on the
identity.

\begin{lemma}
\label{lem:channelproperties} The map $\mathcal{C}_{N}$ satisfies 
\begin{equation}
\mathcal{C}_{N}(I_{\mathcal{P}_{N}})=\frac{N+1}{N+2}I_{\mathcal{P}_{N+1}}.
\label{eq:identityimage}
\end{equation}%
If $\widetilde{\mathcal{C}_{N}}=\frac{N+2}{N+1}\mathcal{C}_{N}$, then 
\begin{equation}
Q_{\widetilde{\mathcal{C}_{N}}(X)}(z)=Q_{X}(z)  \label{eq:symbolpreservation}
\end{equation}%
for every operator $X$ on $\mathcal{P}_{N}$.
\end{lemma}

\begin{proof}
Applying the two Kraus operators to the identity gives (\ref%
{eq:identityimage}) on every basis vector of $\mathcal{P}_{N+1}$. Now let $%
(u,v)$ be a unit vector in $\mathbb{C}^{2}$. Elementary computations using (%
\ref{eq:spincoherent}) give 
\begin{align*}
A_{N,0}^{\ast }\kappa _{N+1}(u,v)& =\left( \frac{N+1}{N+2}\right)
^{1/2}u\,\kappa _{N}(u,v), \\
A_{N,1}^{\ast }\kappa _{N+1}(u,v)& =\left( \frac{N+1}{N+2}\right)
^{1/2}v\,\kappa _{N}(u,v).
\end{align*}%
Consequently, 
\begin{equation*}
Q_{\mathcal{C}_{N}(X)}=\frac{N+1}{N+2}(|u|^{2}+|v|^{2})Q_{X}=\frac{N+1}{N+2}%
Q_{X},
\end{equation*}%
which is (\ref{eq:symbolpreservation}).
\end{proof}

The next result computes the output of a density operator diagonal on the
monomial basis, revealing such output is a density operator also diagonal in
the monomial basis. In short, $\mathcal{C}_{N}$ transforms passive states
into passive states.

\begin{lemma}
\label{lem:passivepreservation} Let $\rho =\func{diag}(p_{0},\ldots ,p_{N})$
in the monomial basis. Then 
\begin{equation*}
\mathcal{C}_{N}(\rho )=\func{diag}(q_{0},\ldots ,q_{N+1}),
\end{equation*}%
where 
\begin{equation}
q_{j}=\frac{(N+1-j)p_{j}+jp_{j-1}}{N+2},\qquad 0\leq j\leq N+1,
\label{eq:birthformula}
\end{equation}%
with $p_{-1}=p_{N+1}=0$. If $p_{0}\geq \cdots \geq p_{N}$, then $q_{0}\geq
\cdots \geq q_{N+1}$.
\end{lemma}

\begin{proof}
Formula (\ref{eq:birthformula}) follows immediately from the two Kraus
operators. For $0\leq j\leq N$, 
\begin{equation*}
(N+2)(q_{j}-q_{j+1})=j(p_{j-1}-p_{j})+(N-j)(p_{j}-p_{j+1})\geq 0.
\end{equation*}%
Thus the output is passive whenever the input is passive. 
\end{proof}

\subsection{Single channel majorization}

The next identity will be used in the proof of Theorem \ref{thm:onestep} and
Lemma \ref{lem:passiveorder}\ below. For $1\leq s\leq N+1$, set 
\begin{equation}
c_{s}=\frac{N+2-s}{N+2}.  \label{eq:cs}
\end{equation}%
Summing (\ref{eq:birthformula}) gives 
\begin{equation}
\sum_{j=0}^{s-1}q_{j}=\sum_{j=0}^{s-2}p_{j}+c_{s}p_{s-1}=c_{s}%
\sum_{j=0}^{s-1}p_{j}+(1-c_{s})\sum_{j=0}^{s-2}p_{j}.
\label{eq:partialpassive}
\end{equation}%
For $s=N+2$, the sum is one.

Our first Theorem says that the output of the density operator $\rho $ via
the channel $\mathcal{C}_{N}$ defined in (\ref{eq:channel}) is always
majorized by the output of the passive rearrangement of the same density
operator.

\begin{theorem}
\label{thm:onestep} For every density operator $\rho $ on $\mathcal{P}_{N}$, 
\begin{equation}
\mathcal{C}_{N}(\rho )\prec \mathcal{C}_{N}(\rho ^{\downarrow }).
\label{eq:onestepmajorization}
\end{equation}
\end{theorem}

\begin{proof}
Let $p_{0}\geq \cdots \geq p_{N}$ be the eigenvalues of $\rho $. Fix $1\leq
s\leq N+1$ and let $R$ be an arbitrary rank-$s$ orthogonal projection on $%
\mathcal{P}_{N+1}$. Since $\mathcal{C}_{N}$ is trace preserving, its adjoint
is unital. Consequently, from (\ref{eq:adjoint-contraction}),    
\begin{equation}
0\leq \mathcal{C}_{N}^{\ast }(R)\leq I_{\mathcal{P}_{N}}.
\label{eq:Bcontraction}
\end{equation}%
Equation (\ref{eq:identityimage}) gives: 
\begin{equation}
\func{Tr}\left( \mathcal{C}_{N}^{\ast }(R)\right) =\func{Tr}\left( R\mathcal{%
C}_{N}(I)\right) =\frac{s(N+1)}{N+2}.  \label{eq:Btrace}
\end{equation}%
Let $\gamma _{0}\geq \cdots \geq \gamma _{N}$ be the eigenvalues of $%
\mathcal{C}_{N}^{\ast }(R)$. The following majorization relation holds: 
\begin{equation}
(\gamma _{0},\ldots ,\gamma _{N})\prec (1,\ldots ,1,c_{s},0,\ldots ,0)=y%
\text{,}  \label{eq:Bmajorized}
\end{equation}%
where the vector on the right has $s-1$ entries equal to $1$, followed by $%
c_{s}$ and then $N+1-s$ zero entries. Denote the entries on the right by $%
y_{j}$. Its partial sums are%
\begin{equation*}
\sum_{j=0}^{m}y_{j}=\left\{ 
\begin{array}{c}
m+1, \\ 
s-1+c_{s},%
\end{array}%
\right. 
\begin{array}{c}
0\leq m\leq s-2 \\ 
s-1\leq m\leq N\text{.}%
\end{array}%
\end{equation*}%
Now, for $0\leq m\leq s-2$, (\ref{eq:Bcontraction}) gives%
\begin{equation*}
\sum_{j=0}^{m}\gamma _{j}\leq m+1=\sum_{j=0}^{m}y_{j}\text{,}
\end{equation*}%
while for $s-1\leq m\leq N$, (\ref{eq:Btrace}) gives%
\begin{equation*}
\sum_{j=0}^{m}\gamma _{j}\leq \sum_{j=0}^{N}\gamma _{j}=\frac{s(N+1)}{N+2}%
=s-1+c_{s}=\sum_{j=0}^{m}y_{j},
\end{equation*}%
which proves the majorization relation (\ref{eq:Bmajorized}). Since the
eigenvalues of $\rho $\ and $\mathcal{C}_{N}^{\ast }(R)$\ are arranged in
decreasing order, the trace inequality for two positive operators and (\ref%
{eq:Bmajorized}) now yield 
\begin{align}
\func{Tr}\left( R\mathcal{C}_{N}(\rho )\right) & =\func{Tr}(\rho \mathcal{C}%
_{N}^{\ast }(R))\leq \sum_{j=0}^{N}p_{j}\gamma _{j}  \notag \\
& \leq \sum_{j=0}^{s-2}p_{j}+c_{s}p_{s-1}.  \label{eq:tracebound}
\end{align}%
By (\ref{eq:partialpassive}), the last expression is the sum of the first $s$
diagonal entries of $\mathcal{C}_{N}(\rho ^{\downarrow })$. Lemma \ref%
{lem:passivepreservation} tells us that these entries are in decreasing
order. Maximizing the left side of (\ref{eq:tracebound}) over all rank-$s$
projections $R$ and using Ky Fan's maximum principle (\ref%
{eq:kyfan-principle}),%
\begin{equation*}
\max_{\substack{ R=R^{\ast }=R^{2} \\ \func{rank}R=s}}\func{Tr}(R\mathcal{C}%
_{N}(\rho ))=\sum_{j=0}^{s-1}\lambda _{j}(\mathcal{C}_{N}(\rho )),\qquad
1\leq s\leq N+1
\end{equation*}
gives the desired $s$th partial-eigenvalue inequality for $1\leq s\leq N+1$.
Since $\func{Tr}(\rho )=\func{Tr}(\rho ^{\downarrow })=1$ and channels are
trace-preserving, the case $s=N+2$ follows from the equality between the two
sums of eigenvalues.
\end{proof}

The proof shows why the finite-dimensional channel is especially well suited
to the problem. No differential inequality is hidden in the argument. The
entire one-step estimate comes from positivity, unitality of the adjoint,
the scalar identity (\ref{eq:identityimage}), and the fact that passive
diagonal entries remain ordered.

\begin{lemma}
\label{lem:passiveorder} Let $p=(p_{0},...,p_{N})$ and $p^{\prime
}=(p_{0}^{\prime },...,p_{N}^{\prime })$ be vectors with decreasing
entries,satisfying%
\begin{equation*}
p_{j}\geq 0,\text{ \ \ \ }p_{j}^{\prime }\geq 0\text{\ \ \ \ \ \ \ \ }%
\sum\limits_{j=0}^{N}p_{j}=\sum\limits_{j=0}^{N}p_{j}^{\prime }=1\text{.}
\end{equation*}%
If $p\prec p^{\prime }$, then the passive outputs determined by (\ref%
{eq:birthformula}) satisfy $q\prec q^{\prime }$.
\end{lemma}

\begin{proof}
For each $1\leq s\leq N+1$, formula (\ref{eq:partialpassive}) expresses the $%
s$th output partial sum as a convex combination of two consecutive input
partial sums. Both input partial sums for $p$ are bounded by the
corresponding sums for $p^{\prime }$. The full output sums are equal to one. 
\end{proof}

\subsection{Proof of Theorem \protect\ref{thm:iterated} (Iterated channel
majorization) }

The assertion is clear for $M=N$. Suppose it holds at degree $M$ and put 
\begin{equation*}
\sigma =\mathcal{C}_{N\rightarrow M}(\rho ),\qquad \tau =\mathcal{C}%
_{N\rightarrow M}(\rho ^{\downarrow }).
\end{equation*}%
Then $\sigma \prec \tau $, and $\tau $ is passive by Lemma \ref%
{lem:passivepreservation}. Theorem \ref{thm:onestep} gives 
\begin{equation*}
\mathcal{C}_{M}(\sigma )\prec \mathcal{C}_{M}(\sigma ^{\downarrow }).
\end{equation*}%
Since $\sigma ^{\downarrow }\prec \tau $ and both states in this last
comparison are passive, Lemma \ref{lem:passiveorder} gives 
\begin{equation*}
\mathcal{C}_{M}(\sigma ^{\downarrow })\prec \mathcal{C}_{M}(\tau ).
\end{equation*}%
Transitivity proves the assertion at degree $M+1$.

\section{Toeplitz operators and $SU(2)$ Husimi majoration}

Introduce the rescaled positive operator 
\begin{equation}
X_{M}(\rho )=\frac{M+1}{N+1}\mathcal{C}_{N\rightarrow M}(\rho )\qquad (M\geq
N).  \label{eq:XM}
\end{equation}%
The next Lemma shows that the Husimi function is invariant under the action
of $X_{M}$\ on $\rho $.

\begin{lemma}
\label{lem:exactlower} For every $M\geq N$, 
\begin{equation}
Q_{X_{M}(\rho )}=Q_{\rho }.  \label{eq:exactlower}
\end{equation}%
Moreover, 
\begin{equation}
0\leq X_{M}(\rho )\leq I_{\mathcal{P}_{M}}.  \label{eq:XMbound}
\end{equation}
\end{lemma}

\begin{proof}
Write $X_{N}(\rho )=\rho $. The recurrence 
\begin{equation*}
X_{M+1}(\rho )=\frac{M+2}{M+1}\mathcal{C}_{M}(X_{M}(\rho ))=\widetilde{%
\mathcal{C}}_{M}(X_{M}(\rho ))
\end{equation*}%
and Lemma \ref{lem:channelproperties} prove (\ref{eq:exactlower}) by
induction. The map $\widetilde{\mathcal{C}}_{M}$ is positive and unital by (%
\ref{eq:identityimage}). Hence it maps the operator interval $[0,I]$ into
itself, which proves (\ref{eq:XMbound}).
\end{proof}

\begin{figure}[t]
\centering
\includegraphics[width=0.95\textwidth]{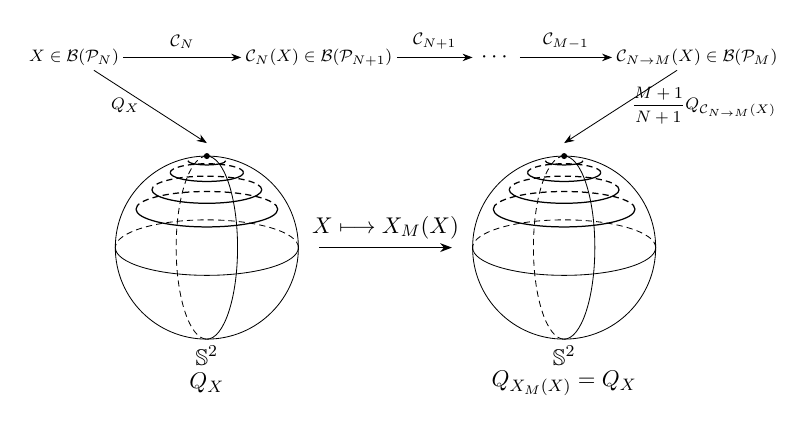}
\caption{Geometric representation of the spherical amplification channel.
Amplification raises the representation degree from $\mathcal P_N$ to
$\mathcal P_M$, while the natural normalization preserves the covariant
symbol and its level curves on the Bloch sphere:
$Q_{X_M(X)}=Q_X$.}
\label{fig:spherical-amplification}
\end{figure}

We have shown that the amplification increases the matrix dimension, while
leaving the observable function on the sphere exactly unchanged. The next
section shows that the eigenvalue distribution of $X_{M}(\rho )$ converges
to the value distribution of that function.

Figure~\ref{fig:spherical-amplification} shows how the amplification channel
maps the input state into a polynomial space of higher dimension, while
 the covariant symbol of the rescaled
channel output remains unchanged: it is exactly the Husimi function of the
input state. Geometrically, this is expressed by preservation of the
level curves on the Bloch sphere.

\subsection{Spherical harmonics and the lower symbol map}

For $0\leq \ell \leq M$, let $\mathcal{Y}_{\ell }$ denote the space of
spherical harmonics of degree $\ell $, regarded as functions on $\mathbb{C}$
by stereographic projection \cite[page 378]{OlverHandbook}. Then $\dim 
\mathcal{Y}_{\ell }=2l+1$ and the spaces $\mathcal{Y}_{\ell }$\ \ are
mutually orthogonal in $L^{2}(\mathbb{C},dm)$. The space of spherical
polynomials of degree at most $M$ is defined by the orthogonal sum: 
\begin{equation*}
\mathcal{V}_{M}=\bigoplus_{\ell =0}^{M}\mathcal{Y}_{\ell }\text{.}
\end{equation*}%
Recall that the lower symbol of an operator $X\in End(\mathcal{P}_{M})$, was
defined in (\ref{lowersymbol}) as%
\begin{equation*}
Q_{X}(z)=\left\langle X\kappa _{M,z},\kappa _{M,z}\right\rangle _{\mathcal{P}%
_{M}}
\end{equation*}%
and write 
\begin{equation*}
X_{j,k}=\left\langle Xe_{k}^{(M)},e_{j}^{(M)}\right\rangle _{\mathcal{P}%
_{M}},\ \ \ \ \ \ 0\leq j,k\leq M.
\end{equation*}

\begin{lemma}
\label{lem:berezincalculus} The lower symbol map $X\rightarrow Q_{X}$ is
injective and a linear isomorphism $End(\mathcal{P}_{M})$ onto $\mathcal{V}%
_{M}$. Moreover, the Berezin transform $B_{M}$ preserves each spherical
harmonic space and satisfies 
\begin{equation}
\mathcal{B}_{M}|_{\mathcal{Y}_{\ell }}=\beta _{M,\ell }I,
\label{eq:berezineigenvalue}
\end{equation}%
where $\beta _{M,\ell }>0$ for $0\leq \ell \leq M$. Thus $B_{M}^{-1}$\ is
defined on $\mathcal{V}_{M}$. Moreover, for fixed $\ell $, 
\begin{equation}
\beta _{M,\ell }\longrightarrow 1\quad \text{as }M\longrightarrow \infty .
\label{eq:betalimit}
\end{equation}
\end{lemma}

\begin{proof}
Expanding $Q_{X}(z)=\left\langle X\kappa _{M,z},\kappa _{M,z}\right\rangle _{%
\mathcal{P}_{M}}$\ in the monomial basis gives 
\begin{equation}
(1+|z|^{2})^{M}Q_{X}(z)=\sum_{j,k=0}^{M}\binom{M}{j}^{1/2}\binom{M}{k}%
^{1/2}X_{jk}z^{j}\overline{z}^{\,k}.  \label{eq:lowersymbolpolynomial}
\end{equation}%
Since the monomials $z^{j}\overline{z}^{\,k}$ are linearly independent, the
lower symbol map $X\rightarrow Q_{X}$ is injective. This map also commutes
with rotations. Under rotations, the space of operators on $\mathcal{P}_{M}$
splits into invariant subspaces of dimensions $1,3,5,...,2M+1,$ and the
subspace of dimension $2\ell +1$ is sent into $\mathcal{Y}_{\ell }$. Since $%
\dim \mathcal{Y}_{\ell }=2\ell +1$\ and the lower symbol map is injective,
each such image equals $\mathcal{Y}_{\ell }$. Consequently,%
\begin{equation*}
\{Q_{X}:X\in End(\mathcal{P}_{M})\}=\mathcal{V}_{M}=\bigoplus_{\ell =0}^{M}%
\mathcal{Y}_{\ell }\text{.}
\end{equation*}%
The kernel in (\ref{eq:berezin}) depends only on spherical distance, so the
Funk--Hecke formula makes $\mathcal{B}_{M}$ scalar on each $\mathcal{Y}%
_{\ell }$. At the north pole that scalar is 
\begin{equation*}
\beta _{M,\ell }=(M+1)\int_{0}^{1}(1-x)^{M}P_{\ell }(1-2x)\,dx.
\end{equation*}%
The shifted Rodrigues formula and $\ell $ integrations by parts give 
\begin{equation*}
\beta _{M,\ell }=\frac{M+1}{\ell !}\frac{M!}{(M-\ell )!}\int_{0}^{1}x^{\ell
}(1-x)^{M}\,dx=\frac{M!(M+1)!}{(M-\ell )!(M+\ell +1)!}=\prod_{j=0}^{\ell -1}%
\frac{M-j}{M+j+2}.
\end{equation*}%
From the last expression, it is clear that $\beta _{M,\ell }>0$ for $0\leq
\ell \leq M$ and that (\ref{eq:betalimit}) holds. 
\end{proof}

\subsection{The spherical classic limit}

We make some observations before the next result, which is the spherical
analogue of the classical limit of de Palma \cite{DePalma2018}. First
observe that if we write%
\begin{equation*}
\mu _{M}=\frac{1}{M+1}\sum\limits_{j=0}^{M}\delta _{\lambda _{j}}(x_{M})
\end{equation*}%
for the empirical spectral distribution of $X_{M}$, then $\mu _{M}$ is a
probability measure on $[0,1]$ and 
\begin{equation*}
\frac{1}{M+1}\func{Tr}\Phi (X_{M})=\int_{[0,1]}\Phi (t)\,d\mu _{M}(t)\text{.}
\end{equation*}%
We will also need the modulus of continuity of a continuous function%
\begin{equation*}
\Phi :[0,1]\rightarrow \mathbb{R}\text{,}
\end{equation*}%
which is defined as, for $s,t\in \lbrack 0,1]$, 
\begin{equation*}
\omega _{\Phi }(\delta )=\sup_{\left\vert s-t\right\vert \leq \delta
}\left\vert \Phi (s)-\Phi (t)\right\vert \text{.}
\end{equation*}%
Since $\Phi $ is continuous in $[0,1]$ , then $\Phi $ is uniformly
continuous, thus%
\begin{equation*}
\lim_{\delta \rightarrow 0}\omega _{\Phi }(\delta )=0\text{.}
\end{equation*}

\begin{lemma}
\label{lem:classical} Let $X_{M}$ be operators on $\mathcal{P}_{M}$ such
that 
\begin{equation*}
0\leq X_{M}\leq I,\qquad Q_{X_{M}}=q,
\end{equation*}%
where $q$ is the lower symbol of an operator on a fixed space $\mathcal{P}%
_{N}$. Then, for every continuous function $\Phi :[0,1]\rightarrow \mathbb{R}
$, 
\begin{equation}
\lim_{M\rightarrow \infty }\frac{1}{M+1}\func{Tr}\Phi
(X_{M})=\lim_{M\rightarrow \infty }\int_{[0,1]}\Phi (t)d\mu _{M}(t)=\int_{%
\mathbb{C}}\Phi (q(z))dm(z).  \label{eq:classicalconv}
\end{equation}
\end{lemma}

\begin{proof}
By Lemma \ref{lem:berezincalculus}, the function $q$ belongs to $\mathcal{V}%
_{N}$. For $M\geq N$, define 
\begin{equation*}
a_{M}=\mathcal{B}_{M}^{-1}q
\end{equation*}%
on this fixed finite-dimensional harmonic space. Equations (\ref%
{eq:berezineigenvalue})--(\ref{eq:betalimit}) imply 
\begin{equation}
a_{M}\longrightarrow q\quad \text{uniformly on the sphere}.
\label{eq:upperconvergence}
\end{equation}%
The lower symbol of $\mathcal{T}_{M}(a_{M})$ is $\mathcal{B}_{M}a_{M}=q$.
The map taking an operator to its coherent-state lower symbol is injective
by Lemma \ref{lem:berezincalculus}. Hence 
\begin{equation}
X_{M}=\mathcal{T}_{M}(a_{M}).  \label{eq:upperrepresentation}
\end{equation}%
It follows that 
\begin{align}
\frac{1}{M+1}\func{Tr}X_{M}^{2}& =\int_{\mathbb{C}}a_{M}(z)Q_{X_{M}}(z)%
\,dm(z)  \notag \\
& =\int_{\mathbb{C}}a_{M}(z)q(z)\,dm(z)\longrightarrow \int_{\mathbb{C}%
}q(z)^{2}\,dm(z).  \label{eq:secondmoment}
\end{align}%
This second moment will be enough to control the convergence for every
continuous $\Phi $ . Consider the spectral decomposition of $X_{M}$, 
\begin{equation*}
X_{M}=\sum_{j=0}^{M}\lambda _{M,j}|\psi _{M,j}\rangle \langle \psi _{M,j}|%
\text{.}
\end{equation*}%
On $\mathbb{C}\times \{0,\ldots ,M\}$, define the probability measure 
\begin{equation*}
d\pi _{M}(z,j)=|\langle \kappa _{M,z},\psi _{M,j}\rangle |^{2}\,dm(z)\text{,}
\end{equation*}%
where counting measure is understood in the discrete variable $j$.\ The $z$%
-marginal is $dm$, while (\ref{eq:resolution}) shows that the $j$-marginal
assigns mass $1/(M+1)$ to each index. For the variables 
\begin{equation*}
\Lambda _{M}(z,j)=\lambda _{M,j},\qquad Y_{M}(z,j)=q(z)\text{,}
\end{equation*}%
we have 
\begin{equation*}
\mathbb{E}_{\pi _{M}}(\Lambda _{M}\mid z)=q(z)=Y_{M}.
\end{equation*}%
Observing that 
\begin{equation*}
\mathbb{E}_{\pi _{M}}(\Lambda _{M}Y_{M})=\mathbb{E}_{\pi _{M}}\left( Y_{M}%
\mathbb{E}_{\pi _{M}}(\Lambda _{M}\mid z)\right) =\mathbb{E}_{\pi
_{M}}(Y_{M}{}^{2})
\end{equation*}

since the variables are real valued, we obtain:%
\begin{eqnarray*}
\mathbb{E}_{\pi _{M}}|\Lambda _{M}-Y_{M}|^{2} &=&\mathbb{E}_{\pi _{M}}\left(
\Lambda _{M}-Y_{M}\right) ^{2}=\mathbb{E}_{\pi _{M}}\Lambda _{M}{}^{2}-2%
\mathbb{E}_{\pi _{M}}\Lambda _{M}Y_{M}+\mathbb{E}_{\pi _{M}}Y_{M}{}^{2} \\
&=&\mathbb{E}_{\pi _{M}}\Lambda _{M}{}^{2}-\mathbb{E}_{\pi _{M}}Y_{M}{}^{2}
\\
&=&\frac{1}{M+1}\func{Tr}X_{M}^{2}-\int_{\mathbb{C}}q^{2}\,dm\longrightarrow
0\text{,}
\end{eqnarray*}%
where we used the limit (\ref{eq:secondmoment}) in the last line. Both
random variables take values in $[0,1]$. If $\omega _{\Phi }$ is the modulus
of continuity of $\Phi $, then for every $\delta >0$, 
\begin{align*}
\mathbb{E}_{\pi _{M}}|\Phi (\Lambda _{M})-\Phi (Y_{M})|& \leq \omega _{\Phi
}(\delta )+2\Vert \Phi \Vert _{\infty }\pi _{M}\{|\Lambda _{M}-Y_{M}|>\delta
\} \\
& \leq \omega _{\Phi }(\delta )+\frac{2\Vert \Phi \Vert _{\infty }}{\delta
^{2}}\mathbb{E}_{\pi _{M}}|\Lambda _{M}-Y_{M}|^{2}.
\end{align*}%
Now let $M\rightarrow \infty $ and then $\delta \downarrow 0$ to yield 
\begin{equation*}
\mathbb{E}_{\pi _{M}}\Phi (\Lambda _{M})-\mathbb{E}_{\pi _{M}}\Phi
(Y_{M})\longrightarrow 0.
\end{equation*}%
The two expectations are precisely the two sides of (\ref{eq:classicalconv}%
). 
\end{proof}

\subsection{Proof of Theorem \protect\ref{thm:sphericaldepalma} on $SU(2)$
Husimi maximization}

We can now pass from matrix majorization to convex maximization of Husimi
functions.

\begin{proof}
\textbf{Step 1. }We first prove the result for continuous convex functions.
Theorem \ref{thm:iterated} and positive rescaling give 
\begin{equation*}
X_{M}(\rho )\prec X_{M}(\rho ^{\downarrow }).
\end{equation*}%
Karamata's inequality therefore yields 
\begin{equation*}
\frac{1}{M+1}\func{Tr}\Phi (X_{M}(\rho ))\leq \frac{1}{M+1}\func{Tr}\Phi
(X_{M}(\rho ^{\downarrow })).
\end{equation*}%
Both operator sequences satisfy the hypotheses of Lemma \ref{lem:classical},
by Lemma \ref{lem:exactlower}. Letting $M\rightarrow \infty $ proves (\ref%
{eq:sphericaldepalma}) when $\Phi $ is continuous.

\textbf{Step 2. }We now show that the continuity assumption on $\Phi $ can
be removed. Indeed, let $\Phi :[0,1]\rightarrow \mathbb{R}$ be any convex
function. For each $n\geq 1$, let $\Phi _{n}$ be the continuous
piecewise-linear function that agrees with $\Phi $ at the points 
\begin{equation*}
0,\frac{1}{n},\frac{2}{n},\ldots ,\frac{n-1}{n},1.
\end{equation*}%
Since $\Phi $ is convex, $\Phi _{n}$ is also convex. Moreover, 
\begin{equation*}
\Phi _{n}(t)\longrightarrow \Phi (t),\qquad 0\leq t\leq 1.
\end{equation*}%
At the endpoints this follows from $\Phi _{n}(0)=\Phi (0)$ and $\Phi
_{n}(1)=\Phi (1)$, while in the interior it follows from the continuity of
every convex function on $(0,1)$. Since $\Phi $ is finite on $[0,1]$, the
functions $\Phi _{n}$ are uniformly bounded. Applying Step 1 to $\Phi _{n}$
gives 
\begin{equation*}
\int_{\mathbb{C}}\Phi _{n}(Q_{\rho }(z))dm(z)\leq \int_{\mathbb{C}}\Phi
_{n}(Q_{\rho ^{\downarrow }}(z))dm(z).
\end{equation*}%
Since $0\leq Q_{\rho },Q_{\rho ^{\downarrow }}\leq 1$, dominated convergence
now gives 
\begin{equation}
\int_{\mathbb{C}}\Phi (Q_{\rho }(z))dm(z)\leq \int_{\mathbb{C}}\Phi (Q_{\rho
^{\downarrow }}(z))dm(z).  \label{eq:convexwithoutcontinuity}
\end{equation}%
Consequently, the conclusion of the preceding theorem holds for every convex
function $\Phi :[0,1]\rightarrow \mathbb{R}$, without a continuity
assumption. Applying the result to $-\Phi $ gives the corresponding reversed
inequality for concave functions.
\end{proof}

\section{Spherical isoperimetric inequalities for $SU(2)$}

\subsection{Concentration operators and Ky Fan's principle}

For a measurable set $\Omega \subset \mathbb{C}$, define the Toeplitz
operator with symbol $1_{\Omega }$, 
\begin{equation}
T_{\Omega }=(N+1)\int_{\Omega }|\kappa _{N,z}\rangle \langle \kappa
_{N,z}|\,dm(z).  \label{eq:concentrationoperator}
\end{equation}%
Then $0\leq T_{\Omega }\leq I$, and 
\begin{equation}
\func{Tr}T_{\Omega }=(N+1)m(\Omega ).  \label{eq:traceconcentration}
\end{equation}%
The operator $T_{\Omega }$\ is therefore trace class and positive, therefore
compact. Let 
\begin{equation*}
\lambda _{1}(\Omega )\geq \cdots \geq \lambda _{N+1}(\Omega )\geq 0
\end{equation*}%
be its eigenvalues. For $1\leq r\leq N+1$, Ky Fan's principle gives 
\begin{equation}
\frac{1}{r}\sum_{j=1}^{r}\lambda _{j}(\Omega )=(N+1)\max_{\substack{ %
P=P^{\ast }=P^{2} \\ \func{rank}P=r}}\int_{\Omega }Q_{P/r}(z)\,dm(z).
\label{KyFan}
\end{equation}%
Indeed, 
\begin{equation*}
Q_{P/r}(z)=\frac{1}{r}\left\langle P\kappa _{N,z},\kappa _{N,z}\right\rangle
_{\mathcal{P}_{N}}\text{,}
\end{equation*}%
and the definition of $T_{\Omega }$ 
\begin{equation*}
Tr\left( PT_{\Omega }\right) =r(N+1)\int_{\Omega }Q_{P/r}(z)dm(z)\text{,}
\end{equation*}%
thus Ky Fan principle indeed gives (\ref{KyFan}).

\subsection{Proof of Theorem \protect\ref{thm:sphericalkyfan}}

\begin{proof}
Theorem \ref{thm:sphericaldepalma} gives, for every $a\in \lbrack 0,1]$, 
\begin{equation}
\int_{\mathbb{C}}(Q_{\rho }-a)_{+}(z)dm(z)\leq \int_{\mathbb{C}}(Q_{\rho
^{\downarrow }}-a)_{+}(z)dm(z).  \label{eq:hinge}
\end{equation}%
This is the exact form needed for concentration on sets of prescribed
measure. The cases $\alpha =0$ and $\alpha =1$ follow from $T_{\varnothing
}=0$ and $T_{\mathbb{S}^{2}}=I$. Hence assume $0<\alpha <1$. The case $r=N+1$
follows from (\ref{eq:traceconcentration}). Assume $1\leq r\leq N$. Let $P$
be an arbitrary rank-$r$ projection and set $\rho =P/r$. Choose $a$ so that,
up to a null boundary, 
\begin{equation*}
D_{\alpha }=\{Q_{\rho _{N,r}}>a\}.
\end{equation*}%
Such an $a$ exists by (\ref{eq:binomialderivative}). Since $m(\Omega
)=\alpha $, 
\begin{align*}
\int_{\Omega }Q_{\rho }(z)dm(z)& \leq \alpha a+\int_{\mathbb{C}}(Q_{\rho
}-a)_{+}(z)dm(z) \\
& \leq \alpha a+\int_{\mathbb{C}}(Q_{\rho _{N,r}}-a)_{+}(z)dm(z) \\
& =\int_{D_{\alpha }}Q_{\rho _{N,r}}(z)dm(z).
\end{align*}%
The middle inequality is (\ref{eq:hinge}). We have thus obtained%
\begin{equation*}
\int_{\Omega }Q_{\rho }(z)dm(z)\leq \int_{D_{\alpha }}Q_{\rho _{N,r}}(z)dm(z)%
\text{.}
\end{equation*}%
Maximizing the left side over $P$ and applying (\ref{KyFan}) shows that%
\begin{equation}
\sum_{j=1}^{r}\lambda _{j}(\Omega )\leq r(N+1)\int_{D_{\alpha }}Q_{\rho
_{N,r}}(z)dm(z)  \label{eqmainbefore}
\end{equation}%
For the cap, the monomials diagonalize the concentration operator, and their
first $r$ elements realize the stated value. The displayed three-line
estimate is the bathtub principle in the form needed here; compare \cite[%
Theorem~1.14]{LiebLoss}. The passive rearrangement of $P/r$ is 
\begin{equation}
\rho _{N,r}=\frac{1}{r}\sum_{n=0}^{r-1}|e_{n}^{(N)}\rangle \langle
e_{n}^{(N)}|.  \label{eq:flatreference}
\end{equation}%
Introduce the spherical radial coordinate 
\begin{equation}
x=\frac{|z|^{2}}{1+|z|^{2}}\in \lbrack 0,1].  \label{eq:xcoord}
\end{equation}%
The measure of the centered cap $\{x<\alpha \}$ is exactly $\alpha $, and 
\begin{equation}
Q_{\rho _{N,r}}(x)=\frac{1}{r}\sum_{n=0}^{r-1}\binom{N}{n}x^{n}(1-x)^{N-n}.
\label{eq:passivehusimi}
\end{equation}%
For $1\leq r\leq N$, 
\begin{equation}
\frac{d}{dx}\sum_{n=0}^{r-1}\binom{N}{n}x^{n}(1-x)^{N-n}=-N\binom{N-1}{r-1}%
x^{r-1}(1-x)^{N-r}.  \label{eq:binomialderivative}
\end{equation}%
Thus $Q_{\rho _{N,r}}$ is strictly decreasing away from the two poles. For $%
r=N+1$, it is the constant $1/(N+1)$.

It remains to show that the right-hand side is the sum of the first $r$
eigenvalues of the concentration operator for the centered cap. This is a
routine calculation. Let 
\begin{equation*}
D_{\alpha }=\left\{ z\in \mathbb{C}:\frac{|z|^{2}}{1+|z|^{2}}<\alpha
\right\} .
\end{equation*}%
Rotational invariance shows that the monomials are orthogonal on $D_{\alpha }
$. Hence 
\begin{align}
\left\langle T_{D_{\alpha }}e_{n}^{(N)},e_{m}^{(N)}\right\rangle _{\mathcal{P%
}_{N}}& =(N+1)\int_{D_{\alpha }}\frac{e_{n}^{(N)}(z)\overline{e_{m}^{(N)}(z)}%
}{(1+|z|^{2})^{N}}\,dm(z)  \notag \\
& =\delta _{n,m}\lambda _{n+1}(D_{\alpha }).  \label{eq:capdiagonal}
\end{align}%
Using polar coordinates and the substitution 
\begin{equation*}
x=\frac{|z|^{2}}{1+|z|^{2}},
\end{equation*}%
we obtain 
\begin{equation}
\lambda _{n+1}(D_{\alpha })=(N+1)\binom{N}{n}\int_{0}^{\alpha
}x^{n}(1-x)^{N-n}\,dx.  \label{eq:caplambda}
\end{equation}%
Thus $T_{D_{\alpha }}$ is diagonal in the monomial basis. Moreover, 
\begin{equation}
\lambda _{n+2}(D_{\alpha })-\lambda _{n+1}(D_{\alpha })=-\binom{N+1}{n+1}%
\alpha ^{n+1}(1-\alpha )^{N-n}<0,  \label{eq:capdecrease}
\end{equation}%
for $0\leq n\leq N-1$. Therefore 
\begin{equation*}
\lambda _{1}(D_{\alpha }),\ldots ,\lambda _{r}(D_{\alpha })
\end{equation*}%
are the first $r$ eigenvalues arranged in decreasing order. Let 
\begin{equation*}
P_{N,r}=\sum_{n=0}^{r-1}|e_{n}^{(N)}\rangle \langle e_{n}^{(N)}|,\qquad \rho
_{N,r}=\frac{1}{r}P_{N,r}.
\end{equation*}%
Since 
\begin{equation*}
\lambda _{n+1}(D_{\alpha })=\left\langle T_{D_{\alpha
}}e_{n}^{(N)},e_{n}^{(N)}\right\rangle _{\mathcal{P}_{N}}=(N+1)\int_{D_{%
\alpha }}\left\vert \left\langle e_{n}^{(N)},\kappa _{N,z}\right\rangle _{%
\mathcal{P}_{N}}\right\vert ^{2}\,dm(z)\text{.}
\end{equation*}

It follows from (\ref{eq:capdiagonal}) that 
\begin{align}
(N+1)\int_{D_{\alpha }}Q_{\rho _{N,r}}(z)\,dm(z)& =\frac{N+1}{r}%
\sum_{n=0}^{r-1}\int_{D_{\alpha }}\left\vert \left\langle e_{n}^{(N)},\kappa
_{N,z}\right\rangle _{\mathcal{P}_{N}}\right\vert ^{2}\,dm(z)  \notag \\
& =\frac{1}{r}\sum_{j=1}^{r}\lambda _{j}(D_{\alpha }).
\label{eq:Husimieigenvaluescap}
\end{align}%
The beta integral in (\ref{eq:caplambda}) also gives 
\begin{equation}
\lambda _{n+1}(D_{\alpha })=\sum_{k=n+1}^{N+1}\binom{N+1}{k}\alpha
^{k}(1-\alpha )^{N+1-k}.  \label{eq:capbinomialsum}
\end{equation}%
Summing over $0\leq n\leq r-1$, we obtain 
\begin{equation}
\sum_{j=1}^{r}\lambda _{j}(D_{\alpha })=r-\sum_{k=0}^{r-1}(r-k)\binom{N+1}{k}%
\alpha ^{k}(1-\alpha )^{N+1-k}.  \label{eq:partialclosed}
\end{equation}%
Combining this identity with (\ref{eqmainbefore}) and (\ref%
{eq:Husimieigenvaluescap}) gives 
\begin{equation*}
\sum_{j=1}^{r}\lambda _{j}(\Omega )\leq r(N+1)\int_{D_{\alpha }}Q_{\rho
_{N,r}}(z)\,dm(z)=\sum_{j=1}^{r}\lambda _{j}(D_{\alpha }).
\end{equation*}%
\hfill 
\end{proof}

\section*{Acknowledgements}
The author is thankful to Michael Speckbacher and Franz Luef, for
discussions around these topics. GPT 5.6 has been used in the editing
process, for the sake of readability and simplicity in exposition and in some auxiliary arguments (in particular
to help achieving a clear presentation of the identification of the channels
in section 3.3), help eliminating typos, mistakes and inconsistencies from
the text, to present the .tex file in an elegant format and to produce the illustration of the action of the spherical amplification channel. This research
was funded in part by the Austrian Science Fund (FWF) through the project
10.55776/PAT8205923 (L.D.A.). For open access purposes, the author has
applied a CC BY public copyright license to any author-accepted manuscript
version arising from this submission.

\end{document}